\documentclass{IEEEtran}[10pt]
\IEEEoverridecommandlockouts

\makeatletter
\newif\if@restonecol
\makeatother

\usepackage{amsmath,graphicx,amsthm,amssymb,epsfig,graphicx,amssymb,url,bm}
\usepackage{threeparttable}
\usepackage{cases}
\usepackage{multirow}
\usepackage{latexsym}
\usepackage{lscape}
\usepackage{arydshln}
\usepackage{balance}
\usepackage{authblk}
\newtheorem{prop}{Proposition}

\usepackage{color}
\providecommand{\keywords}[1]{\textbf{\textit{Index terms---}} #1}

\ifCLASSINFOpdf

\else

\fi

 \makeatletter
\renewcommand\normalsize{%
   \@setfontsize\normalsize\@xpt\@xiipt
   \abovedisplayskip 1\p@ \@plus2\p@ \@minus5\p@
   \abovedisplayshortskip \z@ \@plus3\p@
   \belowdisplayshortskip 6\p@ \@plus3\p@ \@minus3\p@
   \belowdisplayskip \abovedisplayskip
   \let\@listi\@listI}
\makeatother

\begin{document}
\bibliographystyle{IEEEtran}
\graphicspath{{figuresEH/}}
\vspace{-4em}
\title{Efficient Coded Cooperative Networks with Energy Harvesting and Wireless Power Transfer}
%
%
%
%

\author{{Nan Qi, Ming Xiao, Theodoros A. Tsiftsis, Mikael Skoglund, and Huisheng Zhang}
\thanks{
Nan Qi and Huisheng Zhang are with the School of Electronics and Information,  Northwestern Polytechnical University, China (e-mail: naq@kth.se, zhanghusheng@nwpu.edu.cn).

Ming Xiao, Mikael Skoglund are with the School of Electrical Engineering of KTH, Royal Institute of Technology, Stockholm, Sweden (e-mail: \{mingx, skoglund\}@kth.se). Nan Qi is also with KTH.

T. A. Tsiftsis is with the School of Engineering, Nazarbayev University, 010000 Astana, Kazakhstan and with the Department of Electrical Engineering, Technological Educational Institute of Central Greece, 35100 Lamia, Greece (e-mails: theodoros.tsiftsis@nu.edu.kz, tsiftsis@teiste.gr).
}
%
}
\maketitle

\begin{abstract}
The optimum off-line energy management scheme for multi-user multi-relay networks employing energy harvesting and wireless energy transfer is studied. Specifically, the users are capable of harvesting and transferring energy to each other over consecutive transmissions, though they  have no fixed energy supplies. Meanwhile, network coding for the users' messages is conducted at the relays to enable cooperative transmission with source nodes in independent but not necessarily identically distributed (i.n.i.d.)  Nakagami-$m$ fading channels. Therefore, a simultaneous two level cooperation, i.e., information-level  and energy-level cooperation is conducted. The problem of energy efficiency (EE) maximization under constraints of the energy causality and a predefined outage probability threshold is formulated and shown to be non-convex. By exploiting fractional and geometric programming, a convex form-based iterative algorithm is developed to solve the problem efficiently. Close-to-optimal power allocation and energy cooperation policies across consecutive transmissions are found.  Moreover, the effects of relay locations and wireless energy transmission efficiency are investigated and the performance comparison with the current state of solutions demonstrates that the proposed policies can manage the harvested energy more efficiently.

\end{abstract}


\IEEEpeerreviewmaketitle
\section{Introduction}
\subsection{Motivation and Related Works}

In wireless sensor networks (WSNs) or wireless body area networks (WBAN), energy harvesting (EH) is a sustainable approach to prolong the network lifetime. Specifically, 
the EH technique enables the nodes to harvest  energy from nature such as solar, wind, and vibration, and refill their energy-constrained batteries. However, the energy arrival is highly dependent on the environment such as the weather and location. On one hand, the variable weather makes the harvested energy at an individual node intermittent available. On the other hand, distinct locations may lead to some energy-deprived nodes as well as energy-abundant nodes. Both facts cause the inefficiency energy usage of the whole network. Fortunately, wireless energy transfer \cite{SWIPT}-\cite{XLPW} provides a possibility to share the harvested energy among the nodes. That is, it offers extra energy supplement to enhance the data transmission in case that energy harvested from nature is not adequate. All these yield the integrated  energy harvesting and transferring (IEHT) techniques.

IEHT netowrks continues to attract considerable research interest recently \cite{SWIPT}-\cite{DLCS}.  Simultaneous wireless information and power transfer (SWIPT)  was widely considered \cite{YLIU}-\cite{MGL}. Two types of SWIPT protocols are 1) power splitting (PA) \cite{YLIU}-\cite{DWKN}, where the receiver splits the received signal into two parts for decoding information and harvesting energy and 2) time-switching (TS) \cite{HJRZ}, \cite{MGL}, where the receiver switches between decoding information and harvesting energy. These two protocols, which allow receivers harvest energy and receive messages from the same radio frequency (RF) signal have been widely applied in Amplify-and-Forward (AF)/Decode-and-Forward(DF)  based cooperative relaying \cite{YLIU}, orthogonal frequency division multiple access (OFDMA) \cite{ZFTS}, \cite{DWKN} and  single-input multiple-output (SIMO), multiple-input  single-output (MISO) \cite{HJRZ} and multiple-input multiple-output (MIMO) \cite{IKSS} setups. With the  Lagrange dual method, the end-to-end rate, sum-rate or EE are maximized by optimizing the power \cite{YLIU}, \cite{ZFTS}, \cite{DWKN} or time fractions \cite{HJRZ} for harvesting energy or decoding information.   

In the above two protocols,  RF  signals are dually exploited for delivering energy as well as transmitting information.  However, their practical implementation  may not be easily  fulfilled currently given the following limitations of the current state-of-the-art of electronic circuits: 1) the operating power of the energy harvesting unit is much higher than that of the information processing units ($-10$dBm for energy harvesters versus $-60$dBm for information receivers \cite{SWIPT}, \cite{DLCS}); 2) the TS policy requires a strict synchronization process and a non-continuous information transmission; 3) the PS policy requires appropriate PS
circuits that increase the complexity and cost of the hardware, and hardware non-idealities can significant efficiency loss of the PS strategy, as explained in \cite{IKSS}. This raises a demand for \emph{non-overlapping energy harvesting, transferring and information transmitting (Non.o.-IEHT)} techniques. That is, wireless energy transfer is maintained by a separate unit and is independent of the energy harvesting and information transmission, as in \cite{finiteba}, \cite{JX}-\cite{DLCS}. The authors in \cite{finiteba}, \cite{JX}-\cite{DLCS} focused on maximizing the throughput for parallel fading subchannels \cite{WQQ}, two-way cahnnels without \cite{finiteba} or with  relays \cite{twreh} and  MISO network settings \cite{KTAY}. The battery storage can be  finite \cite{finiteba} or  infinite \cite{JX}-\cite{BGOO}. The directional water-filling algorithm obtained via the  Lagrange dual method is widely adopted to obtain the closed-form solutions in \cite{finiteba}, \cite{JX}-\cite{KTAY}. In \cite{DLCS}, the weighted sum-rate maximization problem in energy harvesting  analog network coding (ANC) based TWR  has been investigated. By applying the semi-definite relaxation and successive convex optimization, the optimum beamforming vector and transmitting rates have been determined.  

The above literatures assumed that one relay only assisted one transmitter-destination transmission.  However, in the multi-user multi-relay scenarios, the RF signals broadcasted by each user can be received by  multiple relays; on the other hand, one relay may receive multiple signals from different users and then coordinate their signals.  Therefore, without extra power cost, the system performance can be greatly improved by utilizing the potential of the spatial diversity of distributed relays and  the coordination of the source signals at the relays. One particularly effective way to coordinate source signals is to utilize physical layer network coding (NC)  \cite{MGL}, \cite{NQ}-\cite{20}, which inherently poses an information-level cooperation. In the presence of network coding, user messages are linearly combined over Galois field (GF) to enable sources to cooperate and transmit messages simultaneously. The published works \cite{finiteba}, \cite{JX}-\cite{DLCS} considered the network coding that is operated in GF(2). Recent works illustrated that, if the linear combination is performed  over a large finite field, benefits in terms of diversity order or even energy efficiency can be obtained \cite{17}, \cite{NQ}-\cite{20}. Particularly, to achieve the full diversity order for a group of cooperative users, the concept of maximum diversity network coding (MDNC) was proposed in \cite{19} and \cite{20}. It was shown that an $M$-user $N$-relay network based on MDNC can achieve the full diversity order (i.e., $N+1$ and $N-M+1$ in the presence or  absence  of the direct source node-destination channels, respectively).  It was also proved in \cite{19} that MDNC can provide the network with a larger outage capacity than the dynamic network coding and analog network coding in the high SNR region. 

We consider Non.o.-IEHT in two-hop multi-user multi-relay  systems, where network  coding over high Galois field is also performed, thereby creating a simultaneous two level cooperation, i.e., information-     
 and energy-level cooperation. In this way, the potential of energy efficiency (EE) and wireless resources are expected to be fully exploited. However, to the best of our knowledge, very few works studied energy flow management for such network settings. The authors of the published works \cite{MGL}, \cite{19} and \cite{20} were mainly concerned about the diversity order (i.e., the exponent of SNR in the upper bound); however, the above policies may result in a degraded energy efficiency since they only considered the outage probability performance and energy cost was ignored. In our previous work \cite{NQ}, we presented the energy efficient MDNC networks with Rayleigh fading environment, where power allocation and relay selection were jointly adopted. Nevertheless, the networks are consisted of conventional nodes that cannot harvest or transfer energy. The energy was not fully exploited in the sense that it can neither be shared among the users nor optimized across consecutive transmissions. Moreover, the algorithm is not applicable for the more general Nakagami-$m$ fading environment. An  important and pertinent work on energy harvesting coded networks is \cite{MGL}, where the time-switching based energy transferring protocol has been applied and the time fraction (for harvesting energy or decoding information) was optimized to minimize the  network outage probability over one single time slot. However, it assumed that the outage probability was the same for all inter-user channels and the energy-depletion policy was adopted that  did not allow energy accumulation at the nodes. The algorithm in \cite{MGL} is not feasible in the networks where the inter-channels are independent but not necessarily identically distributed (i.n.i.d.). On the other hand, as we will show in Section V-C,  the energy-depletion policy is not optimal for the consecutive transmission scenarios since the energy was not optimized along the time dimension.  It is thus being observed that, for  network-coded systems employing energy harvesting and transferring, the energy efficient energy flow management which allows extra harvested energy to be accumulated and stored in the batteries for its future usage is still an open problem.

\subsection{Contributions}
In this paper, we study the high Galois field network-coded relaying systems with Non.o.-IEHT techniques. For our considered network model, we focus on careful  management of the energy flow and answering the following questions: 1) To maximize the energy efficiency, how much harvested energy at one specific user should be stored for future usage and how much energy should be transferred  to/obtained from other users in every individual  transmission period? and 2) How to allocate the data transmitting power among the cooperative users and relays such that the EE can be maximized? Specifically, our main contributions are listed as below:

\begin{description}

\item[(1)] We respectively derive the outage probability, energy consumption and EE for the networks that are coded over Galois field. The EE maximizing problem satisfying the energy causality and the pre-defined outage probability threshold is exactly formulated. Different from the similar network coding scenario published in \cite{MGL}, \cite{NQ}-\cite{20}, energy accumulation is allowed at the users rather than depleted  over one transmission; meanwhile,  PA and energy cooperation policies are jointly optimized across consecutive transmissions. That is, energy can flow in time from the past to the future, and in space from one user to the other users. Thus, as we will show later, energy causality constraints take a new form and the underline optimization problem will be completely new and different from the state of arts.

\item[(2)] We consider the Nakagami-$m$ channel, which well models various cellular environments, including the non-line-of-sight (NLOS) and line-of-sight (LOS) channels \cite{APGC}. Moreover, for generality purposes, the channels  are assumed to be i.n.i.d. and the path-loss related to the transmission distance is also incorporated. 

\item[(3)] The optimization problem is shown to be NP-hard. The Lagrange dual method widely adopted in \cite{ZFTS}, \cite{DWKN}, \cite{HJRZ}, \cite{finiteba}, \cite{JX}-\cite{DLCS}, however, is not feasible in our network coding scheme. Instead, to efficiently obtain close-to-optimal solutions, the relaxation and approximation methods are exploited. Finally, a convex form-based iterative algorithm is  developed by combining the geometric programming and non-linear fractional programming.

\item[(4)] The tradeoff between the EE and outage probability is derived. Moreover,  energy cooperation and power allocation policy results are illustrated. The EE gains from NC and energy transferring are also analysed. Additionally, the impacts of the relay locations and the wireless energy transmission efficiency are also investigated.

\end{description}

The rest of the paper is organized as follows. In Section II, we present the system model. EE maximization problem formulation is given in Section III. Then, we reformulate the problem and propose an algorithm in Section IV. The analytical and simulation results are presented in Section V. Section VI concludes this paper.

\section{System Model}

We consider a network with $N$ relays and one destination. There are $M$ users in the network which intend to transmit their independent messages to the destination with the assistance of $N$ full-duplex relays, as depicted in Fig.~\ref{EnergyCooperation_MDNC}. All the nodes are all equipped with a single antenna. It is assumed that there is no direct connection between users and the destination due to the long communication distance or the presence of physical obstacles. 

The destination and relays both have fixed power supplies, while the batteries at the users have to be refilled externally or by the energy transferred from other users. 
There are separate units for wireless energy transferring, energy harvesting and information transmitting such that they  are performed independently and concurrently at one user \cite{finiteba}, \cite{BGOO}, \cite{twreh}.

\begin{figure}[h]
\centering
\includegraphics[width=0.4\textwidth]{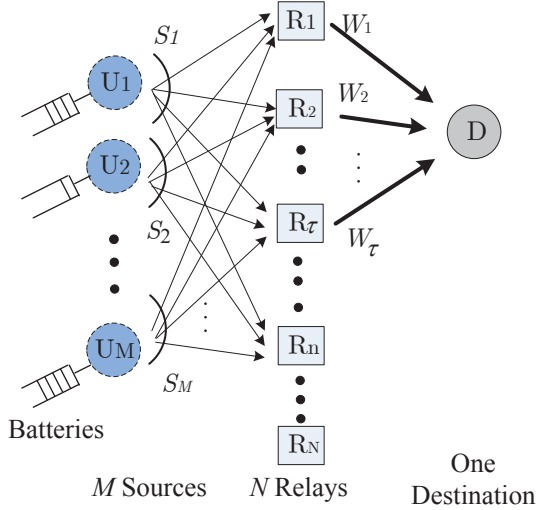}
\caption{Network coded  cooperative transmission.}\label{EnergyCooperation_MDNC}
\end{figure}
 
In what follows, we illustrate the channel model, information transmission, energy harvesting and cooperation models, respectively.

\subsection{Channel Model}

Slowly varying flat fading channels that follow Nakagami-$m$ distribution are considered. We note that  Nakagami-$m$ fading is a general channel model in the sense that variations of the severity of any fading channel can be expressed as Nakagami-$m$ distribution by changing the fading parameter, $m$, from $1/2$ to $+ \infty$ \cite{APGC}. Let $U_i$, $i\in \{1, 2, \cdots, M\}$, and $R_j$, $j\in \{1, 2, \cdots, N\}$, represent the $i$th user and $j$th relay, respectively. Then  $U_i$-$R_j$ channel coefficient is represented as
\begin{equation}
h_{ij}  = |\Upsilon _{ij} |\sqrt {d_{ij}^{ - \beta _{ij} } }\exp (j\varphi _{h_{ij} } ),\label{hijdef}
\end{equation}
where  $h_{{ij }}$ is the channel gain that combines the path-loss and Nakagami-$m$ fading; $|\Upsilon _{ij} |$ is the Nakagami-$m$ fading envelope; $\varphi _{h_{ij} }$ is the phase of the $U_i$-$R_j$ channel that is assumed uniformly distributed over the range of $[0,2\pi)$; $d_{ij }^{ - \beta_{ij }}$ denotes the path loss; ${d_{ij }}$ is the distance and ${\beta_{ij }}$ is the channel path loss exponent.

The probability distribution function (pdf)  of $|\Upsilon _{ij} |$ can be given as \cite{MNakagami}
\begin{equation}
f_{|\Upsilon _{ij}|} (x) = \frac{{2m^m x^{2m - 1} }}{{\Gamma (m)\Omega _{h_{ij} }^m }}\exp ( - \frac{{x^2 }}{{\Omega _{h_{ij} }^{} }}), x > 0,\label{pdfhijk}
\end{equation}
where $\Gamma ( \cdot )$ is the gamma function; ${{\Omega _{h_{ij} }^{} }}$ is the average
channel gain represented as ${{\Omega _{h_{ij} }^{} }}=\mathbb{E}[
{|\Upsilon _{ij} |^2 }]$ and $\mathbb{E}\{\cdot\}$ is the expectation operator.

Similar to $h_{{ij }}$ in  \eqref{hijdef}, the $R_j$-destination channel coefficient, denoted as ${g_{j}}$, also combines the path-loss and Nakagami-$m$ channel fading, i.e.,
\begin{equation}
g_{j}  = |\Upsilon _{j}|\sqrt {d_{j}^{ - \beta _{j} } }\exp (j\varphi _{g_{j} } ),\label{gjdef}
\end{equation}
where  $|\Upsilon _{j}|$, $\varphi _{g_{j} }$, $d_{j}$ and $\beta _{j}$ are parameters for the $R_j$-destination channel and denote  the channel  fading envelope,  channel phase, distance and path loss exponent, respectively. $|\Upsilon _{j}|$ also follows the Nakagami-$m$ distribution. The average
channel gain is represented as ${{\Omega _{g_{j} }^{} }}=\mathbb{E}[
{|\Upsilon _{j}|^2 } ]$. 

We assume that perfect channel state information (CSI) is available at the receivers, while the transmitters only have the knowlege of ${{\Omega _{h_{ij} } }}$ and ${{\Omega _{g_{j} } }}$ ($\forall i, j$). In the  i.n.i.d. fading environment, ${{\Omega _{h_{ij} } }}$ and ${{\Omega _{g_{j} } }}$ ($\forall i, j$) may differ from each other.

\subsection{Information Transmission Scheme}

All nodes operate in time division multiple access (TDMA), which is also adpoted in \cite{MGL} and \cite{17}. Thus, there is no interference among information transmissions. 

As shown in Fig. \ref{EnergyCooperation_MDNC}, the whole transmission consists of two hops.

\begin{description}
  \item[\textit{1)\;The First Hop: User-relay Transmission}]
\end{description}

The message of $U_i$ ($\forall i$) is denoted as $S_i$. Suppose that all user messages are of the same length\footnote{We note that this assumption is made for simplifying illustration. The system model can be  extended to general cases where different users may have different message lengths. {More specifically, if different users have different message lengths, we can divide the messages into shorter ones such that the lengths of shorter messages are the same and some users have more messages while some have fewer messages. Then, the users with fewer messages may not participate  in  all transmission rounds.}}, denoted as $|S|$. Additionally, we assume that all users and relays transmit information with a fixed rate $\alpha_0$ bits per second\footnote{Our model and algorithm are also applicable for different fixed rates on different channels. The rates affect the values of the data transmitting time and  outage probability. However, different rates have no impact in the convexity of $\mathbf{P3}$ presented in Section IV. Hence, the analysis and proposed scheme are still feasible.}.   Take $U_i$ as an example. $S_{i}$ is first protected by channel
coding and then modulated into a unit power-signal,
denoted as $X(S_{i})$. 	Then $X(S_{i})$ is broadcast to all relays, which takes $T=|X(S_{i})|/\alpha_0$ seconds ($|\cdot|$ means the number of bits in $X(S_{i})$). 

$R_j$ receives the signal from $U_i$ as follows:
\begin{equation}
F_{ij}= h_{{ij }}\sqrt { p_{i }} X(S_{i}) + z_{{ij }}, \label{eq1}
\end{equation}
where ${ p_{i }}$  is the transmitting power for the channel codeword at $U_i$; $z_{ij } \sim {\mathcal{N}}(0,N _{0,ij}B)$ denotes the AWGN; $N _{0,ij}$ is the one-sided power spectral density and $B$ is the bandwidth.

The achievable rate for the channel between $U_i$ and $R_j$ is
\begin{equation}
C_{ij }= B{\log _2}(1 + \frac{|{h_{ij }}{|^2}p_{i }}{{N_{0,ij}}B}), \label{eq0}
\end{equation}
where $|{h_{ij }}|$ is the amplitude of ${h_{ij }}$.
An outage event occurs in the $U_i$-$R_j$ channel when the fixed data transmission rate is larger than the Shannon capacity \cite{19},  i.e., 
\begin{equation}
C_{ij } < \alpha_0. \label{outcreterion}
\end{equation}
If no outage event happens in the $U_i$-$R_j$ channel, $R_{j}$  will decode $F_{ij }$ into ${S_{i}}$. In this way, $R_{j}$ tries to obtain all source messages, i.e., $\{S_1, S_2, \cdots, S_M\}$.

\begin{description}
  \item[\textit{2)\;The Second Hop:  Relay-Destination Transmission}]
\end{description}

The following notations will be used in our following description.

${\Theta}$:  the index set of all the  relays.

${{{\Phi _{n}}}}$: Suppose in the first hop, $n$ relays succeed in receiving and decoding all the user messages. Their index set is ${{{\Phi _{n}}}}$. Note that ${n=0}$ means that no relay receives and decodes all the user messages.

$\psi_{\tau}$:  Suppose in the second hop, $\tau$ relays manage  to forward  messages to the destination.  ${\tau}=0$ means no relay forwarding messages to the destination.

Clearly, $\tau \leq n \leq N$ and $\psi_{\tau} \subseteq {{{\Phi _{n}}}} \subseteq {\Theta}$.


If $R_j$ fails to decode any user message, it will not forward messages. Otherwise, if it can decode all user messages,  a network coding scheme based on pre-defined MDNC coding coefficients will be applied. A network codeword ${W_{j}}$ is generated at $R_{j}$ by the linear combination of ${S_{1}},\,{S_{2}},\, \cdots \,\,,{S_{M}}$ over a finite field, i.e.,
\begin{equation}
{W_{j}} = \mathop \boxplus \limits_{i = 1}^M \gamma_
{ij}{{S_{i}} }, \forall i \in \{1,2,\cdots,M\}, \forall j \in {{{\Phi _{n}}}},
 \nonumber 
 \end{equation}  
where ``$\boxplus$" is the addition operation in the finite field; $\gamma_
{ij}$ is the global encoding kernel for ${{S_{i}} }$ at relay $R_{j}$. $\gamma_
{ij}$ constitutes the transfer matrix ${\mathbf{G}_{M \times N}}$ corresponding to MDNC\footnote{Since the design of MDNC encoding and decoding are not our main points, we skip their design details.}
\begin{equation}
{\mathbf{G}_{M\times N}} = \left( {\begin{array}{*{20}{c}}
{{\gamma_
{11}}}&{{\gamma_
{12}}}& \ldots &{{\gamma_
{1N}}}\\
{{\gamma_
{21}}}&{{\gamma_
{22}}}& \ldots &{{\gamma_
{2N}}}\\
{}& \ldots & \ldots &{}\\
{{\gamma_
{M1}}}&{{\gamma_
{M2}}}& \ldots &{{\gamma_
{MN}}}
\end{array}} \right).
\label{Tmatric} 
\end{equation}
${\mathbf{G}_{M \times N}}$  is row full rank \cite{19}. 

Before being forwarded to the destination, ${W_{j }}$ ($\forall j \in {{{\Phi _{n}}}}$) is first protected by channel coding and then modulated into a unit-power signal, denoted as  
$X({W_{j}})$. Note that with network coding, we have $|{W_{j}}|=|{S}|$. Correspondingly, every transmission in the second hop also takes $T$ seconds. 
 
At the destination, the signal from $R_j$  is received  as
\begin{equation}
  F({ W_{j}})={g _{j}}\sqrt {{{p'_{j}}}} X({W_{j}}) + z_{j}, \forall {j} \in {\psi_{\tau }}, \label{eq5}
\end{equation}
where ${{{p'_{j}}}}$ is the transmitting power at relay $R_j$; $z_{j} \sim {\mathcal{N}}(0,{N_{0,j}}B)$ is the noise term; ${N_{0,j}}$ is the power spectral density of noise and ${g_{j}}$ is the channel coefficient.

Finally, the destination obtains  ${S_{1 }}$, ${S_{2 }}$,  $\cdots$, ${S_{M }}$ jointly from $\{F({ W_{j}}), j\in{\psi_\tau}\}$ by network decoding. In effect, the row full rank property of ${\mathbf{G}_{M \times N}}$ guarantees that $\{{S_{1 }}, {S_{2 }},  \cdots, {S_{M }}\}$ can be recovered at the destination as long as 
\begin{equation}
\tau \geq M, \label{taulessM}
\end{equation}
otherwise, none of the user messages can be obtained and we claim an outage event happens for all user message transmissions.

A transmission period is defined as the duration in which all the $M$ users complete one cycle broadcasting in the TDMA scheme, which lasts $MT$ seconds. In total, $K$ consecutive transmission periods are considered. Let $k$ ($k=0, 1, \cdots, K$) represent the index of the transmission period. Take $U_1$ as an example, as illustrated in Fig. \ref{figEnergyarrival}, the $k$th transmission period corresponds to $t \in [(k - 1)MT, kMT)$. The second hop in the $k$th transmission period can simultaneously proceed with the first hop of the $(k+1)$th transmission period. Thus, the duration for one cycle ``first hop + second hop" transmission is equivalently calculated as $MT$ seconds.

\subsection{Energy Harvesting and Cooperation Model}

As shown in Fig. \ref{figEnergyarrival}, in any transmission period, the incoming energy at one user is  either harvested externally  or obtained from other users via wireless energy transfer (e.g., in the $2$nd transmission period, $U_1$ obtains energy from other users). Correspondingly, the user can either consume energy for the data transmission or  transfer energy to the other users (e.g., in the $K$th transmission period, $U_1$ transfers energy to other users). 

\begin{figure}[h]
\centering
\includegraphics[width=0.4\textwidth]{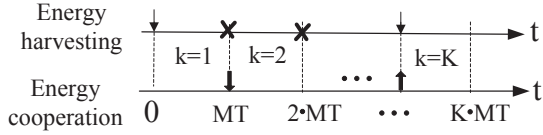}
\caption{Energy harvesting and transfer models for $U_1$. The thin arrow marks that energy is harvested from the external environment while ``$\times$" means no harvested energy. The bold and upward arrow records that  energy is transferred to  other users during the corresponding transmission period, while the bold and downward arrow represents that energy is provided by other users.}\label{figEnergyarrival}
\end{figure}

In our energy harvesting model, we consider the off-line policy \cite{twreh}, in the sense that the harvested energy amount and harvesting time are known (or can be precisely predicted) to all nodes in advance. We assume that $U_i$ ($\forall i$) harvests $Eu_{i,k}$ ($\forall k$) joules of energy from the external environment when $t \in [(k - 2)MT,(k - 1)MT)$. Obviously, $Eu_{i,k}$ can be consumed in the $k$th and later transmission periods. Note the users do not necessarily harvest energy at the same time. That is, $Eu_{i,k}$ ($\forall i$, $\forall k$) can be zero. 

We further describe the energy transfer morel. Suppose that in the $k$th transmission period, $U_i$ transfers ${{E_{i\to i', k}}}$ joules energy to ${U_{i'}}$. We define an energy transfer matrix ${\mathbf{E}_{I \to I',k}}$ ($\forall k$), as
\vspace{+0.5em}
\begin{equation}
{\mathbf{E}_{I \to I',k}} = \left( {\begin{array}{*{20}{c}}
0&{{E_{1\to2, k}}}& \ldots &{{E_{1\to M, k}}}\\
{{E_{2\to1, k}}}&0& \ldots &{{E_{2\to M, k}}}\\
{}& \ldots & \ldots &{}\\
{{E_{1\to M, k}}}&{{E_{2\to M, k}}}& \ldots &0
\end{array}} \right).
\nonumber 
\vspace{+0.7em}
\end{equation}
Note that all the diagonal elements of ${\mathbf{E}_{I \to I',k}}$ are zero, since one user does not transfer energy to itself. The set of ${\mathbf{E}_{I \to I',k}}$ ($\forall k$) is denoted as ${\mathbf{E}_{I \to I'}}$.

The wireless energy transmission efficiency is denoted as  $\eta < 1$. Then, $U_i$ receives $\eta \sum_{i' = 1}^M {E_{i' \to i,{k}}} $ joules energy  from the other cooperative source nodes during the $k$th transmission period. Correspondingly, $U_i$ totally transfers $\sum_{i' = 1 }^M {E_{i\to i',{k}}}$ joules energy to other users.  We claim that either $\eta \sum_{i' = 1}^M {E_{i' \to i,{k}}} $ or $\sum_{i' = 1 }^M {E_{i\to i',{k}}}$  must be zero. This is because if one user simultaneously obtains energy from other users and transfers energy to other users,  unnecessary energy loss will happen due to the wireless energy transmission inefficiency.

The energy evolution process is depicted as below. Let ${Eu_{i,k}^{ava}}$ denote the available energy for \emph{data transmission} in the $k$th transmission period. The available energy for the $1$st transmission period is described as
\vspace{+0.5em}
\begin{align}
Eu_{i,1}^{ava} =  E{u_{i,0}}+ E{u_{i,1}}+ \eta \sum\limits_{i' = 1}^M {E_{i' \to i,1}} - \sum\limits_{i' = 1 }^M {E_{i \to i',1}},\forall i,
\label{1EnuncS}
\end{align}
where $Eu_{i,0}$ stands for the initial energy storage, which is set as $0$ without loss of generality;   
 $E{u_{i,1}}$ is the energy harvested externally before the first transmission period.

For $k\geq 2$,  ${Eu_{i,k}^{ava}}$ evolves as
\vspace{+0.5em}
\begin{align}
Eu_{i,k }^{ava}=&Eu_{i,k-1}^{ava}- p_{i,{k-1}}T +E{u_{i,k }}  \hfill\nonumber\\
&+\eta \sum\limits_{i' = 1}^M {E_{i' \to i,{k}}}- \sum\limits_{i' = 1 }^M {E_{i \to i',{k}}} , \label{EuncS}\hfill
\end{align}
where $\eta \sum_{i' = 1}^M {E_{i' \to i,{k}}} $ is the possibly obtained energy  from the other cooperative sources during the $k$th transmission period while $\sum_{i' = 1 }^M {E_{i \to i',{k}}}$ indicates the energy transferred to other users  and $p_{i,{k-1}}$ denotes  the  data transmission power during the $(k-1)$th transmission period, respectively.

Without loss of generality, we assume that ${d_{ij }}$, $\Omega _{{h_{ij }}}^2$, ${ \beta_{ij }}$ and  $N _{0,ij}$ and their corresponding parameters in the second hop are the same in all $K$ transmission periods. Thus, we can drop the index $k$ in the above mentioned parameters to ease the following notations.

\section{EE Maximization Problem}

In this section, we respectively formulate the expressions for the EE, total consumed energy and outage probability. Following that, the EE maximization problem is finally presented.

\subsection{Energy Efficiency}

EE is evaluated as the expected number of successfully transmitted information bits, $\mathbb{E}[L]$, divided by the total consumed energy ${E_{tot}}$, i.e., 
\begin{equation}
U_{EE}=\frac {\mathbb{E}[L]}{{{E_{tot}}}}.\label{EEMDNC}
\end{equation}
\vspace{+0.3em}

As we have illustrated in \eqref{taulessM},  the network decoder at the destination either recovers all source messages or cannot decode any of them, the outage probabilities of all users are the same. Let   
 ${{\Pr}_{out,k}}$ denote the outage probability for all the users in the $k$th transmission period. Thus,  over $K$ consecutive transmission periods,  $\mathbb{E}[L]$ can be expressed as \cite{NQ}
\vspace{+0.5em}
\begin{equation}
\mathbb{E}[L]=\sum\limits_{k = 1}^K{M{{\alpha_0}T{(1-{{\Pr}_{out,k}}) }}}.
\end{equation}

In the following, we give the expressions for ${{\Pr}_{out,k}}$, ${{{E_{tot}}}}$ and their corresponding constraints in different transmission periods, respectively.

\subsection{Total Consumed Energy}

The total consumed energy includes the energy used for data transmission and wasted during the energy cooperation, which is denoted by
\vspace{+0.5em}
\begin{equation}
{E_{tot}}= \sum\limits_{k = 1}^K\left( {\sum\limits_{i = 1}^M {{p_{i,k}T} + (1 - \eta )\sum\limits_{i = 1}^M {\sum\limits_{i' = 1 }^M {{E_{i \to i',k}}}  + } \sum\limits_{j = 1}^N {p'_{j,k}} T} } \right),\label{Etot}
\end{equation}
where the term $(1 - \eta )\sum_{i = 1}^M(\cdot)$ in \eqref{Etot} represents the overall  energy loss incurred by the wireless energy transmission inefficiency during the wireless energy transfer.

In our model, an infinite-sized battery capacity  at the user is assumed, which has also been adopted in \cite{JX}, \cite{WQQ}. Specifically, a super-capacitor can be implied to store the incoming energy. Since the energy that has not yet arrived cannot be consumed ahead of time due to the energy causality, it is required that the consumed energy amount cannot exceed the available amount. Correspondingly, we formulate the power control constraints as
\vspace{+0.5em}
\begin{align}
&0 < p_{i,k}T\leq {Eu_{i,k}^{ava}}, \forall i, k, \label{const_ps}\hfill\\
&0 < p_{i,k}\leq  p_{max}, \forall i, k, \label{const_ps2}\hfill\\
&0 \leq p'_{j,k} \leq p_{max}, \forall j, k, \label{const_pr}
\end{align}
\vspace{+0.5em}
where $p_{max}$ is the maximum transmitting power.

\subsection{Outage Probability}

We  first  give the exact and approximated outage probability expression for one individual channel,  based on which the outage probability of the whole network is derived.

\vspace{+0.3em}
\begin{description}
  \item[\textit{1)\;Outage Probability of One Individual Channel}]
\end{description}
 \vspace{+0.3em}
 
For the Nakagami-$m$ fading channel, take the $U_i$-$R_j$ channel as an example, the outage probability  can be calculated according to \eqref{eq0} and \eqref{outcreterion}. Specifically, we have 
\vspace{+0.5em}
\begin{align}
 {{\Pr}_{e,{ij,k}}}&={\Pr}\{ C_{ij}  < {\alpha_0}\}={\Pr}\{ |{h_{ij}|^2} < \frac {({2^{{{\alpha_0} \mathord{\left/
 {\vphantom {R B}} \right.
 \kern-\nulldelimiterspace} B}}} - 1){N_{0,ij}}B}{ p_i}\}  \hfill\nonumber\\
 &= 1 - \int_\frac {({2^{{{\alpha_0} \mathord{\left/
 {\vphantom {R B}} \right.
 \kern-\nulldelimiterspace} B}}} - 1){N_{0,ij}}B}{ p_i}^{+\infty} f(x_{ij}){d{x_{ij}}}\hfill\nonumber\\
&= \Gamma \left( {m,\frac{{m(2^{\alpha _0 /B}  - 1)N_{0,ij} B}}{{d_{ij}^{ - \beta _{ij} } \Omega _{h_{ij} } p_{i,k} }}} \right) \Gamma (m)^{ - 1},  \label{nakagaPeijk}
\end{align}
where $\Gamma (a,b) = \int_0^{ b } {x^{a - 1} } \exp ( - x)dx$ is the upper incomplete gamma function and $\Gamma (a) = \int_0^{ + \infty } {x^{a - 1} } \exp ( - x)dx$ is the complete gamma function.

The outage probability of one individual channel in \eqref{nakagaPeijk} is  not tractable mathematically since $p_{i,k}$ is not isolated but contained in the gamma function. However, the incomplete gamma function can be well approximated as \cite{ZWGG}
\begin{equation}
 \Gamma (a,b) \approx  (1/a)b^a \nonumber
\end{equation}
for small $b$. This approximation offers one method in isolating $p_{i,k}$ from the gamma function. Specifically, \eqref{nakagaPeijk} can be approximated as
\begin{align}
 {{\Pr}_{e,{ij,k}}}
 & \approx \left( {\frac{{m(2^{\alpha _0 /B}  - 1)N_{0,ij} B}}{{d_{ij}^{ - \beta _{ij} } \Omega _{h_{ij} } p_{i,k} }}} \right)^m \Gamma (m + 1)^{ - 1}  \nonumber\\
 & = c_{ij} p_{i,k}^{ - m},  \label{nakagaPeijk_appro}
\end{align}
where 
\begin{equation}
 c_{ij} = \left( {\frac{{m(2^{\alpha _0 /B}  - 1)N_{0,ij} B}}{{d_{ij}^{ - \beta _{ij} } \Omega _{h_{ij} } }}} \right)^m \Gamma (m + 1)^{ - 1}>0.\label{c_ij} 
\end{equation}

Similarly, the outage probability of the $R_j$-D channel can be given as
\begin{equation}
 {{\Pr}_{e,j,k}} \approx  c_{j} {(p'_{j,k})}^{ - m},  \label{nakagaPejk_appro}
\end{equation}
where 
\begin{equation}
c_j  =\left( {\frac{{m(2^{\alpha _0 /B}  - 1)N_{0,j} B}}{{d_j^{ - \beta _j } \Omega _{g_j } }}} \right)^m \Gamma (m + 1)^{ - 1} >0.  \label{c_j}
\end{equation}

As can be seen from \eqref{c_ij} and \eqref{c_j}, $c_{ij}$ and $c_{j}$ combine all the channel paramenters.  Increasing $c_{ij}$ or $c_{j}$ will lead to larger outage probability of one indiviual channel. Thus, larger $c_{ij}$ and $c_{j}$ represent worse channel conditions.

\vspace{+0.3em}
\begin{description}
  \item[\textit{2)\;Outage Probability of the Whole Network}]
\end{description}
\vspace{+0.3em}
 
As we illustrated in \eqref{taulessM}, an outage event happens when $\tau \leq M$. In the following, we focus on deriving the probability that $\tau \leq M$.

Suppose in the $k$th transmission period, $n$ relays succeed in receiving all the source messages.  An outage event happens in the following two cases in terms of $n$. In case ${A_k}$, {$n< M$}.  User messages cannot be recovered no matter how the second hop proceeds. In case $B_k$, $n \geq M$. An outage event  happens when the number of relays forwarding the codewords to the BS in the second hop is smaller than $M$. 

We  denote the probability that case ${A_k}$ and ${B_k}$ repectively happening as $\Pr\{A_k\}$ and $\Pr\{B_k\}$. Since cases ${A_k}$ and ${B_k}$ are independent, then the outage probability for the  whole network can be calculated as
\begin{equation}
{\Pr} _{out,k}=\Pr\{A_k\}+\Pr\{B_k\}. \label{completePout}
\end{equation}  
 
We have respectively  formulated $\Pr\{A_k\}$ and $\Pr\{B_k\}$ as \eqref{exactpoutA} and \eqref{exactPoutB} in \cite{NQ}, where the nodes are not capable of  harvesting or transferring energy  and the channels follow Rayleigh fading. 
\begin{equation}
\Pr\{A_k\} =\sum\limits_{n = 0}^{M - 1} \sum\nolimits_{{\Phi _{n,k}}} {\left( {\prod\limits_{{j} \in {\Phi _{n,k}}} {{\rho _{j,k}}} {\prod _{{j} \in {\Theta}\backslash {\Phi _{n,k}}}}(1 - {\rho _{j,k}})} \right)}. \label{exactpoutA}
\end{equation}

\begin{figure*}[ht]
\begin{small}
 \begin{flalign}
& \Pr\{B_k\}= \sum\limits_{n = M}^{N} \left({\sum\nolimits_{{\Phi _{n,k}}} {( {\prod\limits_{{j} \in {\Phi _{n,k}}} {{\rho _{j,k}}} {\prod _{{j} \in {\Theta}\backslash {\Phi _{n,k}}}}(1 - {\rho _{j,k}})} )} \cdot \sum\limits_{\tau  = 0}^{M - 1} \sum\nolimits_{{\psi _{\tau,k} }} {( {\prod\limits_{{j} \in {\psi_{\tau,k} }} {(1 - {{\Pr }_{e,j,k}})} {\prod _{{j} \in {\Phi _{n,k}}\backslash {\psi _{\tau,k} }}}{{\Pr}_{e,j,k}}} )}}\right).\label{exactPoutB}&
\end{flalign}
\end{small}
\hrulefill
\end{figure*}

We note that  $\sum\nolimits_{{\Phi _{n,k}}} {(\beta)}$ in \eqref{exactpoutA} and \eqref{exactPoutB} represents the sum of $\beta$ when ${{\Phi _{n,k}}}$ is in different cases. ${{\Phi _{n,k}}}$ consists of $n$ relays randomly chosen from  $N$ relays in the $k$th transmission period, including $C_N^n$ cases. ${{\psi _{\tau,k} }}$ consists of $\tau$ relays randomly chosen from  $n$ relays in the $k$th transmission period, including $C_n^{\tau}$ cases. In \eqref{completePout} and \eqref{exactpoutA}, ${\rho _{j,k}}$ measures the probability that $R_j$ manages to receive all the $M$ user messages in the $k$th transmission period. It is evaluated by \cite{NQ}
\begin{equation}
 {\rho _{j,k}} = \prod\limits_{i = 1}^M {(1 - {{\Pr}_{e,ijk}})}.\label{ro}
\end{equation}

For the energy harvesting and cooperation scenario, where the channels follow Nakagami-$m$ fading distributions, we can obtain ${{\Pr}_{out,k}}$ by substituting ${\Pr}_{e,ij,k}$ and ${\Pr}_{e,j,k}$  into \eqref{completePout}-\eqref{exactPoutB}.  

Our objective is to maximize the EE across $K$ transmission periods by jointly optimizing $p_{i,k}$, $p'_{j,k}$ and ${{E_{i \to i',k}}}$ ($\forall i, j, k$) according to the harvested energy and the channel parameters, including $Eu_{i,k}$, $d_{ij }$, $\beta_{ij }$, $\Omega _{{h_{ij }}}$, $d_{j }$, $\beta_{j }$ and $\Omega _{{g_{j }}}$ ($\forall i, j, k$). The optimization problem can be formulated as
\begin{align}
& \mathbf{P1}: \mathop {\max }\;U_{EE} \nonumber\hfill\\
&s.t.\;\;  \eqref{const_ps}-\eqref{const_pr},\nonumber \hfill\\
&\quad\;\; {{\Pr}_{out,k}}\leq {{\Pr}_{out,0}}, \forall k,\label{Poutconstr}
\end{align}
where ${{\Pr}_{out,0}}$ is the predefined outage probability threshold for every transmission period. We note that the outage probability threshold may vary in different transmission periods. To ease the notations, we set them as the same value, denoted as ${{\Pr}_{out,0}}$.

\section{Problem Transformation and Solving}

The key challengings in solving $\mathbf{P1}$ stem from the following facts.

Although $\Pr_{out}$ in $\mathbf{P1}$ represents the exact outage probability, it
consists of multiple exponential items. Note the coefficients
of exponential items are positive and negative constants that
alternately appear. This makes the outage probability constraint and the objective function in $\mathbf{P1}$ neither in their concave nor convex forms.  

On the other hand, the widely adopted Lagrange duality method \cite{24} in \cite{ZFTS}, \cite{DWKN}, \cite{HJRZ}, \cite{finiteba}, \cite{JX}-\cite{KTAY} is not applicable in our network coding scenario since the product forms of $p_{i,k}$ ($\forall i, k$)  and $p'_{j,k}$ ($\forall j$)  make the equations obtained via KKT conditions very complicated to be solved. Hereby, the closed-form solutions are hard to achieve. The Brute-force algorithm is also infeasible even for small $M$ and $N$. That is because for the network with $M$ users and $N$ relays, there are $Num_{Var}=(M-1)(M-1)K+MK+NK$
variables to be determined in total, including $(M-1)(M-1)K$ energy transferring variables, $MK$  power allocation variables of $M$ users and $NK$ power allocation variables of $N$ relays. 

In the sequel,  we exploit the relaxation and approximation methods, which alleviate the optimization difficulties substantially. First, the outage probability is converted into its geometric programming form. Then, we covert the objective function and energy causality constraints into their convex forms, thereby finally converting the primal optimization problem into a standard convex one. The details are given in the following.

\subsection{Transformation of the Outage Probability}

In the case of low outage probability threshold ${{\Pr}_{out,0}}$, both ${{c_{ij}}/p_{i,k}^m }$ and $c_j /(p'_{j,k} )^m$ are required to be small according to \eqref{completePout}-\eqref{exactPoutB}. The outage probability constraint can be satisfied  if the transmitting power is allocated appropriately and  $c_{ij}$ and $c_j$ (specifically, $\alpha_0/B$ and  noise power) are small. With small values of ${{c_{ij}}/p_{i,k}^m }$ and $c_j /(p'_{j,k} )^m$, we can derive the following approximations
\begin{equation}
1 - {\Pr}_{e,ij,k}  = 1 - c_{ij} p_{i,k}^{ - m}  \approx \exp ( - c_{ij} p_{i,k}^{ - m} )\sim 1,\label{1_Pr}
\end{equation}
and
\begin{equation}
{\rho _{j,k}}\mathop \sim \limits^{\eqref{ro}} 1. \label{Poutapp}
\end{equation}

Moreover, since $\mathop {\lim }\limits_{x \to 0} \;  1 - \exp ( - x) = x$, we have
\begin{align}
 1 - \rho _{j,k}  &= 1 - \prod\limits_{i = 1}^M {(1 - {\Pr}_{e,ij,k})} \nonumber \\ 
  &\mathop  = \limits^{\eqref{1_Pr}} 1 - \prod\limits_{i = 1}^M {\exp ( - c_{ij} p_{i,k}^{ - m} )}  \nonumber\\ 
  & \approx \sum\limits_{i = 1}^M {c_{ij} p_{i,k}^{ - m} }.\label{noout1}
\vspace{+0.7em}
\end{align}

By substituting \eqref{1_Pr}-\eqref{noout1} into \eqref{exactpoutA} and \eqref{exactPoutB}, we obtain the tight  approximations for $\Pr\{A_k\}$ and  $\Pr\{B_k\}$, which are repectively given in \eqref{appropoutA1} and \eqref{eq56}.
\begin{equation}
\Pr\{A_k\} \approx \sum\limits_{n = 0}^{M - 1} {\sum\nolimits_{{\Phi _{n,k}}} {\left( {\prod\limits_{{j} \in {\Theta}\backslash {\Phi _{n,k}}} {\sum\limits_{i = 1}^M {\frac{{{c_{ij}}}}{{{p_{i,k}^m}}}} } } \right)} } \label{appropoutA1}
\end{equation}

\begin{figure*}[ht]
\begin{small}
 \begin{flalign}
 &\Pr\{B_k\}\approx   \sum\limits_{n = M}^N {\left( {(\sum\nolimits_{{\Phi _{n,k}}} {(\prod\limits_{{j} \in {\Theta}\backslash {\Phi _{n,k}}} {(\sum\limits_{i = 1}^M {\frac{{{c_{ij}}}}{{{p_{i,k}^m}}}} } )} )) \cdot (\sum\limits_{\tau  = 0}^{M - 1} {\sum\nolimits_{{\Phi _{n,k}}} {(\prod\limits_{{j} \in {\Phi _{n,k}}\backslash {\psi _\tau }} {\frac{{{c_j}}} {(p'_{j,k} )^m }  } } ))} } \right)}.
 & \label{eq56}
\end{flalign}
\end{small}
\end{figure*}

Furthermore, we introduce two new variables, i.e., ${{\tilde p}_{i,k}}$ and ${{\tilde p'}_{j,k}}$ as below
\begin{equation}
{p_{i,k}} = {e^{{{\tilde p}_{i,k}}}},\quad \; { {{p'_{j,k}}}} ={e^{ {{\tilde p'}_{j,k}}}}.\label{eq489}
\end{equation}
By substituting ${{{\tilde p}_{i,k}}}$ and ${{\tilde p'}_{j,k}}$ into \eqref{appropoutA1} and \eqref{eq56}, we have \eqref{finalappropoutA} and \eqref{eq222}. 

\begin{equation}
\Pr\{A_k\} \approx \sum\limits_{n = 0}^{M - 1} {\sum\nolimits_{{\Phi _{n,k}}} {\left( {\prod\limits_{{j} \in {\Theta}\backslash {\Phi _{n,k}}} {\sum\limits_{i = 1}^M {{c_{ij}}{e^{ - m{{\tilde p}_{i,k}}}}} } } \right)} }. \label{finalappropoutA}
\end{equation}
\vspace{+0.3em}

\begin{figure*}[ht]
\begin{small}
 \begin{flalign}
&\Pr\{B_k\}\approx \sum\limits_{n = M}^N {\left( {(\sum\nolimits_{{\Phi _{n,k}}} {(\prod\limits_{{j} \in {\Theta}\backslash {\Phi _{n,k}}} {(\sum\limits_{i = 1}^M {{{c_{ij}}e^{ - m{{\tilde p}_{i,k}}}}} } )} )) \cdot (\sum\limits_{\tau  = 0}^{M - 1} {\sum\nolimits_{{\Phi _{n,k}}} {(\prod\limits_{{j} \in {\Phi _{n,k}}\backslash {\psi _\tau }} {{e^{ - m{{\tilde p'}_{j,k}}}}} } ))} } \right)}. & \label{eq222}
\end{flalign}
\end{small}
\hrulefill
\end{figure*}

As can be seen,  both  \eqref{appropoutA1} and \eqref{eq56} are given in the geometric programming forms of $p_{i,k}$ and $p'_{j,k}$ while \eqref{finalappropoutA} and \eqref{eq222} are in the sum-exponential forms of ${{{\tilde p}_{i,k}}}$ and ${{\tilde p'}_{j,k}}$. With  all the coefficients being positive,  ${{\Pr}_{out,k}}$ is finally approximated to its convex form. Then the objective function can be    
dealt with parametric method based on the fractional programming theory.

\subsection{Transformation of the Objective Function}

In the sequel, we apply Dinkelbach's method  \cite{OFDM_nonl} to transform the fractional problem into its subtractive form. The following proposition is provided.
\begin{prop}  \label{Theorem1}
The PA and energy cooperation  policies can achieve the maximum energy efficiency \begin{equation}
q^*=\max \{{U_{EE}}\},\nonumber 
\end{equation}
if and only if
\begin{align}
V( q^ *, \mathcal{\tilde P}^*, {\mathbf{E}_{I \to I'}^*}) &= \max_{q\geq0} \{{M{\alpha _0}T\sum\limits_{k = 1}^K (1-{{\Pr}_{out,k}}) } - { q }{E_{tot}})\} \nonumber\\
&= 0, \label{non}
\end{align}
where $q^*$ is the maximum EE, ${\mathbf{E}_{I \to I'}^*}$ is the optimal energy cooperation policy, $\mathcal{P}^*$ is the set of optimum solutions of ${p_{i,k}}$,  ${p'_{j,k}}$ ($\forall i, j, k$), while $\mathcal{\tilde P}^*$ is the set of optimum solutions of ${{\tilde p}_{i,k}}$,  ${{\tilde p'}_{j,k}}$ ($\forall i, j, k$).
\end{prop}

According to Proposition \ref{Theorem1}, we reformulate $\mathbf{P1}$  as $\mathbf{P2}$.
\begin{align}
&\mathbf{P2}: \max V ={{M{\alpha _0}T\sum\limits_{k = 1}^K (1-{{\Pr}_{out,k}}) }} - {q}{E_{tot}} \label{con_objective}\\
&s.t.\;\;  \eqref{Poutconstr},\nonumber\hfill\\
&\quad\;\;   0< p_{i, k}T=e^{{{\tilde p}_{i, k}}}T\leq {Eu_{i,k}^{ava}}, \forall i, k,\label{eq63eu} \hfill\\
&\quad\;\;   0< p_{i, k}=e^{{{\tilde p}_{i, k}}}\leq  p_{max}, \forall i, k,\label{eq63} \hfill\\
&\quad\;\;   0\leq  {{p'_{j, k}}}={e^{ {{\tilde p'}_{j, k}}}} \leq p_{max}, \forall j, k.\label{eq631a}
\end{align}

We focus on  finding $q^ *$,  $\mathcal{\tilde P}^*$ and ${\mathbf{E}_{I \to I'}^*}$ to satisfy
\begin{equation}
\max \{V(q^ *, \mathcal{\tilde P}^*, {\mathbf{E}_{I \to I'}^*})\}=0. 
\end{equation}
In  Dinkelbach's method, $q$ is iteratively updated in every iteration; meanwhile,  $V(q^ *, \mathcal{\tilde P}^*, {\mathbf{E}_{I \to I'}^*})$ is judged whether it converges to a given tolerance. If not, $q$ is updated and we repeat the maximization problem until it converges or reaches the maximal iterations.

Note that with given $q$ in every iteration, we have
 \begin{equation}
 \max \;V \Leftrightarrow \min \; V'=-V+MK{\alpha_0}{T},\label{Vdef}
\end{equation}

Correspondingly, $\mathbf{P2}$ is equivalently transformed into
$\mathbf{P3}$ as shown below,
\begin{align}{}
&\mathbf{P3}:  \min\;V'( q^ *, \mathcal{\tilde P}^*, {\mathbf{E}_{I \to I'}^*})  \nonumber\\
& s.t.\quad\eqref{Poutconstr}, \eqref{eq63eu}-\eqref{eq631a}. \nonumber
\end{align}

According to \eqref{non} and \eqref{Vdef}, the optimum solution of $\mathbf{P3}$ must satisfy
\begin{equation}
\min \{V'( q^ *, \mathcal{\tilde P}^*, {\mathbf{E}_{I \to I'}^*})\} =MK{\alpha_0}T. \label{zeromaxi}
\end{equation}

We first provide a proposition for $\mathbf{P3}$.
\begin{prop}  \label{Theorem2}
Given $q$, $\mathbf{P3}$ is  jointly convex   with respect to (\emph{w.r.t})  ${{{\tilde p}_{i,k}}}$, ${{{\tilde p'}_{j,k}}}$ and ${E_{i \to i',k}}$, $\forall i, i', j, k$. Efficient interior-point method can be applied to obtain its optimum solution.
\end{prop}
\begin{proof}
Proof is provided in the Appendix.
\end{proof}

We summarize the overall procedure to solve $\mathbf{P1}$ in Algorithm $1$.

\begin{figure}[h]
\centering
\includegraphics[width=0.5\textwidth]{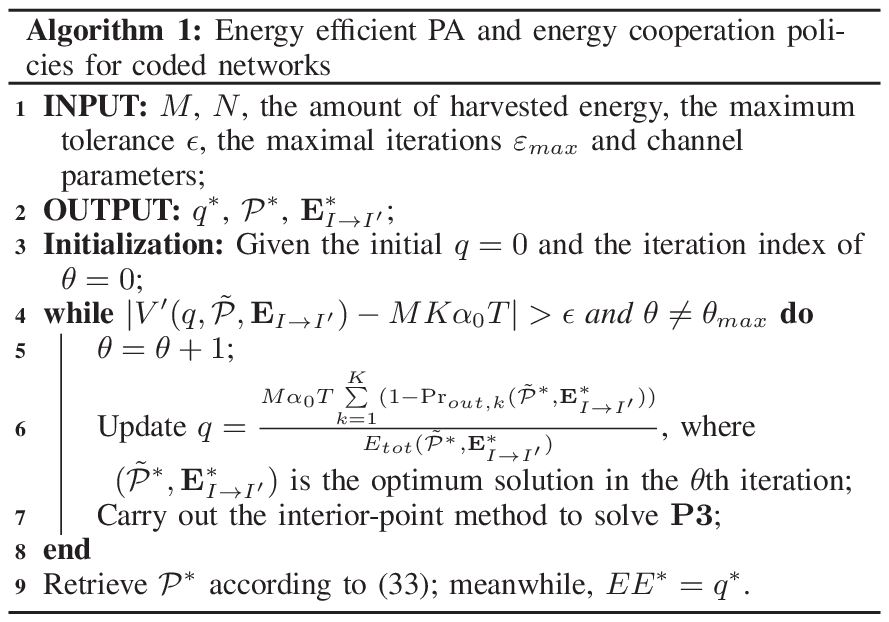}
\end{figure}

\emph{Complexity Analysis:}  With the Brute-force algorithm, the complexity is $\mathcal{O}(\upsilon ^{Num_{Var}})$, where $\upsilon$ is the iteration time for  one variable and determined by the step size. As we can see, the complexity with the Brute-force algorithm increases exponentially with $M$, $N$  and $K$. In contrast, with Dinkebach's method, the iteration time for $q$ is limited \cite{OFDM_nonl}. Furthermore, with the interior-point method applied, the complexity will be $\mathcal{O}(\mathcal{C}_1\mathcal{C}_2)$, where  $\mathcal{C}_1={(Num_{Var} + Z + 1)^{1/2}}$, $\mathcal{C}_2=(Num_{Var} + 1){{Z}^2} + {{Z}^3} + {Num_{Var}^3}$ and $Z$ is the total number of exponential terms in the objective and constraints \cite{interior_complexity}.  It can be found that $\mathcal{O}(\mathcal{C}_1\mathcal{C}_2)$ is a polynomial in $M$, $N$, and $K$.

\emph{Model Extension:} {Our energy harvesting and transferring cooperative networks model can also be extended to more general scenarios.} For example,  relays  may also be capable of harvesting energy externally  and transferring energy to users or  other  relays. In this case,  when formulate the EE maximization problem, we can regard the relays as users and let $U_{M+j}$ represent $R_{j}$. In other words, $U_i$ denotes one relay rather than an user if $i \in \{M+1, M+2, \cdots, M+N\}$. Specifically, the following two minor changes are needed. First, we rewrite \eqref{const_pr}, i.e., the power constraint at relay $R_j$ as
\begin{align}
&0 < p_{i,k}T\leq \min \{{Eu_{i,k}^{ava}}, p_{max}T\}, \hfill \\
&\quad\quad\quad i \in \{M+1, M+2, \cdots, M+N\},  \forall k, \nonumber\hfill
\end{align}
where $p_{M+j,k}$ and ${Eu_{M+j,k}^{ava}}$ are the transmitting power and available energy at $R_j$ in the $k$th transmission period, respectively. 
Second,  ${Eu_{i,k}^{ava}}$ ($i \in \{1, 2, \cdots, M+N\}$) is not only related with the incoming/outgoing energy transferred from/to other users but also the relays. Then the expressions for available energy at sources and relays can be simply obtained by changing ``$M$" in \eqref{1EnuncS} and \eqref{EuncS} into ``$M+N$".  Algorithm $1$ is still feasible for the extended scenarios. On
the other hand, if a part of users are not energy harvesting nodes or cannot transfer energy to other users/relays, then we only
need to delete the energy causality constraint for them. Algorithm $1$ is also applicable to such kind of scenarios.

\section{Numerical Results}

In what follows, we will present numerical results. Energy arrivals are generated randomly and independently. Their specific values are shown in Fig.~\ref{OptimalPolicy}. We simulate the process for $10$ consecutive transmission periods.  We  assume that $M=2$,  $N=4$, $B=125K$Hz, $\alpha_0=10^5$ bits per second and $p_{max}=20$ watts.  We normalize $T$ as $1$. The following  randomly generated values are also assumed,
\begin{small}
\begin{align}
  &{{\mathbf{\Omega }}_{\mathbf{h}}} = \left[ {\begin{array}{*{20}{c}}
  {2.5646}&{1.7520 }&{2.1684 }&{0.5798} \\
  {2.3024}&{0.4753}&{3.5462}&{0.3904}
\end{array}} \right], \hfill \nonumber\\
  &{{\mathbf{\Omega }}_{\mathbf{g}}} = \left[ {\begin{array}{*{20}{c}}
  { 3.3120}&{1.1286}&{0.3284}&{0.7821}
\end{array}} \right], \hfill \nonumber\\
  &{{\mathbf{d}}_{\mathbf{h}}} = \left[ {\begin{array}{*{20}{c}}
  {1565}&{765.2}&{1704.2}&{1530.4} \\
  {1979.6}&{1831}&{720.6}&{173.6}
\end{array}} \right], \hfill\nonumber \\
 &{{\mathbf{\beta}}_{\mathbf{h}}} = \left[ {\begin{array}{*{20}{c}}
  {2.5570}&{2.9150 }&{2.3152}&{3.0143} \\
  {3.0938 }&{2.1298}&{2.6412}&{2.9708}
\end{array}} \right], \hfill \nonumber\\
  &{{\mathbf{d}}_{\mathbf{g}}} = \left[ {\begin{array}{*{20}{c}}
  {471.7}&{1045.7}&{902.2}&{1079.4}
\end{array}} \right], \hfill \nonumber \\
&{{\mathbf{\beta}}_{\mathbf{g}}} = \left[ {\begin{array}{*{20}{c}}
  {2.6103}&{3.2838}&{1.8435 }&{2.3515}
\end{array}} \right], \hfill \nonumber\\
  &{{\mathbf{N}}_{{\mathbf{0}},{\mathbf{h}}}} = 10^{-15}\left[ {\begin{array}{*{20}{c}}
  {0.126}&{0.07}&{0.54}&{0.006} \\
  {0.002}&{0.002}&{0.429}&{1.096}
\end{array}} \right], \hfill \nonumber\\
  &{{\mathbf{N}}_{{\mathbf{0}},{\mathbf{g}}}} = 10^{-15}\left[ {\begin{array}{*{20}{c}}
  {0.3799}&{0.7243}&{0.0265}&{0.1225}
\end{array}} \right]. \hfill \nonumber
\end{align}
\end{small}
Note that ${{\mathbf{\Omega }}_{\mathbf{h}}}$ and ${{\mathbf{\Omega }}_{\mathbf{g}}}$ denote the variance matrices of the average channel gain; ${\mathbf{{N_{0,h}}}}$ and  ${\mathbf{{N_{0,g}}}}$ represent the power spectrum density matrices, which are measured in Watts/Hz. ${\mathbf{{d _h}}}$ and ${\mathbf{{d _g}}}$ represent the distance matrices which are measured in meter. ${\mathbf{{\beta _h}}}$ and ${\mathbf{{\beta_g}}}$ represent the path-loss exponent matrices.  Specifically, elements at the $i$th row and the $j$th column of ${\mathbf{{\Omega _h}}}$, ${\mathbf{{d _h}}}$, ${\mathbf{{\beta_h}}}$ and ${\mathbf{{N_{0,h}}}}$ correspond to the parameters for the $S_i-R_j$ channel. Element at the $j$th columns in ${\mathbf{{\Omega _g}}}$, $\mathbf{{d _g}}$, ${\mathbf{{n _g}}}$ and ${\mathbf{{N_{0,g}}}}$ corresponds to the parameters for the $R_j$-destination channel.

\subsection{Optimal Policy Illustration}

We first take the scenario when $m=1$  (i.e., Rayleigh fading channel), $\eta=0.6$ and $\Pr_{out,0}=6 \times 10^{-7}$ as example. The cumulative harvested energy and the optimal power policies are depicted in Fig.~\ref{OptimalPolicy}.

For the cumulative harvested energy curves, the rising height at the beginning of every transmission period represents the amount of newly harvested energy that can be consumed in that period. Zero-rising height implies that no energy is harvested. From Fig. \ref{OptimalPolicy}, we observe that $U_2$ harvests sufficient energy from the external environment  while $U_1$ harvests small amount of energy and suffers from energy deficiency. Especially in the $2$nd and $3$rd transmission periods, as can be seen, no energy is harvested at $U_1$.

\begin{figure}[h]
\centering
\includegraphics[width=0.5\textwidth]{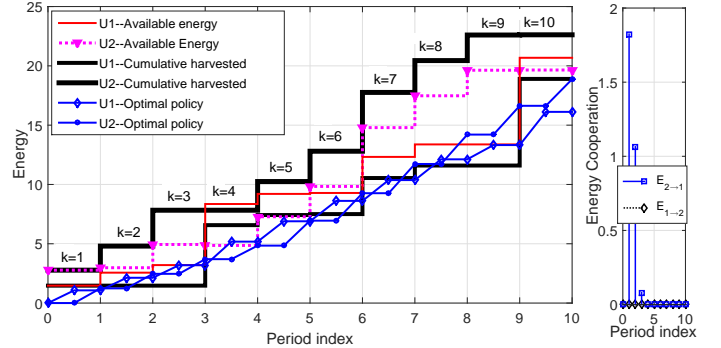}
\caption{Optimal policy and cumulative harvested energy; $\eta=0.6$; $m=1$.}
\centering\label{OptimalPolicy}
\end{figure}

To compensate for $U_1$ such that NC cooperative transmission can be carried out, the energy cooperation policy is adopted. Fig. \ref{OptimalPolicy} shows that $U_2$ respectively transfers $1.8240$, $1.0642$ and $0.08$ joules energy to $U_1$ during the $2$nd, $3$rd and $4$th transmission periods. Take the second transmission period as an example. Since $U_2$ transfers energy to $U_1$ during the $2$nd transmission period, the available energy at $U_2$ is smaller than the cumulative harvested amount. In contrast, due to the additional incoming energy from $U_2$, the available energy at $U_1$ exceeds the cumulative harvested amount.

For the optimal policy curve, the slope of one line segment represents the transmitting power in the corresponding transmission period. Zero-slope represents that no energy is consumed and no transmission proceeds.  We can observe that due to the TDMA transmission scheme, the zero-slope line segments in the optimal policy curves of $U_1$ and $U_2$ alternately appear. Moreover, as can be seen, the optimal power policy curves of  $U_1$ and  $U_2$ in Fig. \ref{OptimalPolicy} are not higher than the available energy curves due to the energy causality constraint. 

Additionally, in some transmission periods, the optimal consumed energy amount is not necessarily the same with the available amount. In other words, the energy is not depleted and some is saved and will be consumed later for the sake of maximizing the EE. Take $U_1$ as an example, the available energies are not used up till the end of the $3$rd and $9$th transmission periods. Similar conclusions can be obtained for $U_2$. Specially, though no harvested energy or cooperation energy from $U_1$ in the $10$th transmission period, data transmission still proceeds at $U_2$ and the outage probability threshold is satisfied, which benefits from its cumulative harvested energy in the prior transmission periods, which is not achievable with the policy in \cite{MGL}.

\subsection{Impacts of the Relay Locations}
To investigate the impact of relay locations on the EE, we fix the distance between the users and destinations but move the relays  between the users and  the destination. To be specific, the distance between $U_i$ ($\forall i$),  and $R_j$ ($\forall j$), is changed into $({d_{{ij}}}+\Delta)$ meters while the distance between $R_j$ and the destination is reduced into $({d_{j}}-\Delta)$ meters, where $\Delta$ is the shifting distance of one relay.

\begin{figure}[h]
\centering
\includegraphics[width=0.4\textwidth]{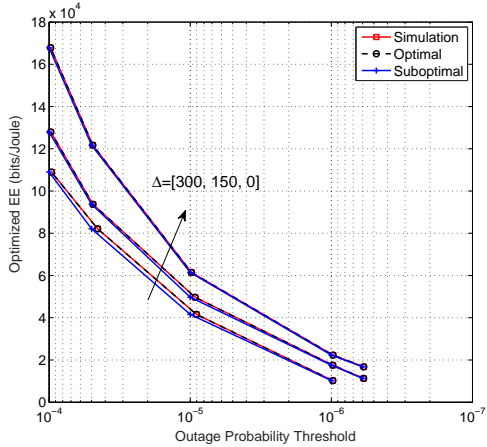}
\caption{Effects of relay locations; $\eta=0.6$; $\Pr_{out,0}=[10^{-4}, 5\times10^{-5}, 10^{-5}, 10^{-6}, 6\times10^{-7}]$; $m=1$.}
\centering\label{effects_relaylocations}
\end{figure}

In Fig. \ref{effects_relaylocations} and Fig. \ref{effects_relaylocationsM3}, the tradeoff curves between ${U_{EE}}$ and the predefined outage threshold, $\Pr_{out,0}$ for $\Delta=0, 150, 300$ cases are depicted. The channels are assumed to be either LOS (e.g., $m=3$) or NLOS (e.g., $m=1$) ones. The numerical results are obtained by  Algorithm $1$, the Brute-force algorithm  and simulations, respectively. The simulation results are obtained by respectively averaging the outage probability, total consumed energy and EE over $10^9$ random realizations of the fading channels. As can be seen, for both NLOS and LOS channel scnearios, their analytical results obtained from Algorithm $1$ closely match the results from the Brute-force algorithm  and simulation, especially in the low $\Pr_{out,0}$ region where higher SNR is needed. All these show  that the analytical results obtained by Algorithm $1$ are valid.  

We can also observe that EE decreases with the pre-defined outage probability threshold, which implies that the decrease in the outage probability threshold can cost significant EE penalty.

\begin{figure}[h]
\centering
\includegraphics[width=0.4\textwidth]{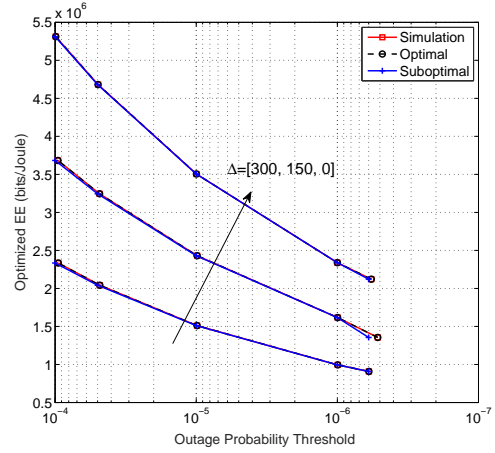}
\caption{Impacts of relay locations; $\eta=0.6$; $\Pr_{out,0}=[10^{-4}, 5\times10^{-5}, 10^{-5}, 10^{-6}, 6\times10^{-7}]$; $m=3$.}
\centering\label{effects_relaylocationsM3}
\end{figure}

In Fig. \ref{effects_relaylocations}, the gaps among the tradeoff curves demonstrate the EE loss resulting from  increasing $\Delta$ in the NLOS channel environment. It can be noticed that $12\%$ and  $25\%$ EE losses are respectively generated for the cases when $\Delta=150$ and $300$ meters.  $\Pr_{out,0}=6 \times 10^{-7}$ is even not achievable when $\Delta=300$. Similar conclusions can also be obtained for the LOS scenario, as shown in Fig. \ref{effects_relaylocationsM3}. It is because the first hop transmission is dominant in the two-hop transmission scheme. The increase in the transmission distance of the first hop deteriorates the outage probability performance, which needs more energy in the second hop to compensate and results in lower EE.

In Fig. \ref{effects_relaylocationsM3},  we obtian  the optimum EE when  $\eta=0.6$ in the LOS channel environment (e.g., $m=3$). As can be observed, its EE is around 20 times that of NLOS scenarios.  Additionally,  in contrast to the NLOS scenarios, the LOS results show that, even for the strictest outage probability requirement (i.e., when  $\Pr_{out,0}=6 \times 10^{-7}$), no energy cooperation is needed among the users. Thus, energy loss is avoided during the wireless energy transferring.  This advocates the rationale since if $m$ increases, the channels become more advantageous for data transmission. In other words, less power is needed to meet a specific target outage probability level, which is also clearly revealed in \eqref{eq222}. Hereby, the increase of $m$ results in a higher EE. 

\subsection{Impacts of the Energy Transmission Efficiency}

In Fig. \ref{effects_transmissioneffiencyd0}, the EE curves for the scenario when $\Delta=150$ and $\eta=0.2$, $0.6$ and $1$ are plotted. It is clear that for the same $\Pr_{out,0}$, more EE losses are caused when $\eta$ takes a smaller value. Moreover, $\Pr_{out,0}=6 \times 10^{-7}$ is not achievable when $\eta=0.2$ due to the significant energy losses during the energy cooperation. Note that three curves overlap when $\Pr_{out,0}=1 \times 10^{-4}$ due to the fact that no energy cooperation is needed to satisfy the outage probability threshold. Such numerical results give references on the system parameter settings.

\begin{figure}[h]
\centering
\includegraphics[width=0.4\textwidth]{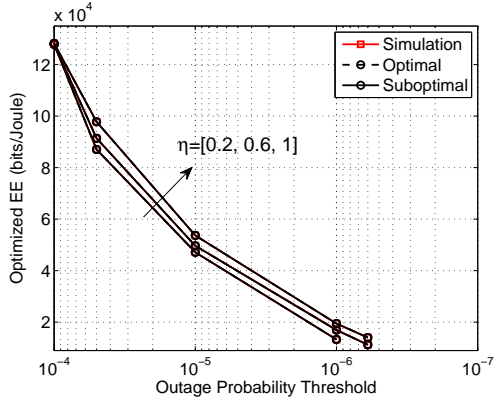}
\caption{Impacts of wireless energy transmission efficiency; $\Delta=150$; $\Pr_{out,0}=[10^{-4}, 5\times10^{-5}, 10^{-5}, 10^{-6}, 6\times10^{-7}]$; $m=1$.}
\centering\label{effects_transmissioneffiencyd0}
\end{figure}

\subsection{Performance comparison of different transmission schemes}

For comparison, in Fig.~\ref{Performance_Comparison}, we provide close-to-optimal results obtained by our proposed algorithm for the scenario without network coding (NoNC) \cite{MMB}. In the NoNC scenario, decode-and-forward (DF) relaying protocol is adopted at  $N$ relays.  It is shown  that the EE of  the NC scenario is more than $30\%$ higher than that of the NoNC scenario, which demonstrates that considerable EE gains can be achieved with NC.

\begin{figure}[h]
\centering
\includegraphics[width=0.4\textwidth]{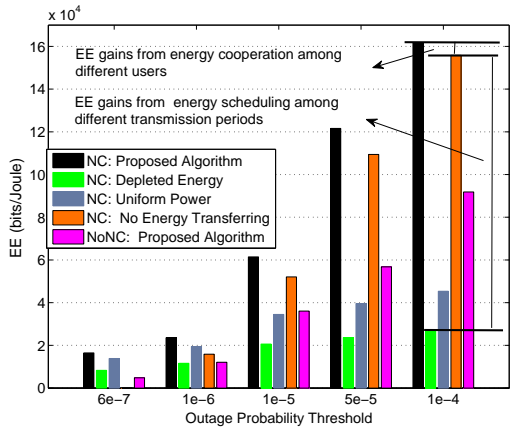}
\caption{EE performance comparison for different transmission strategies; $\eta=0.6$; $\Delta=0$; $m=1$.}
\centering\label{Performance_Comparison}
\end{figure}

{Moreover, for the coded scenario,  EE obtained with another three transmission strategies is also provided. First, in the ``No Energy Transferring" scheme,  there is no energy cooperation among the users but the energy can be schuduled over different transmission periods. The EE gap between ``No Energy Transferring"  and our proposed algorithm $1$ indicates the EE gain obtained from the energy cooperation. Specifically, $46.5\%$, $18\%$, $11\%$ and $3.7\%$ EE gains can be obtained when  $\Pr_{out,0}=10^{-6}$, $10^{-5}$, $5\times10^{-5}$, $10^{-4}$, respectively. It can be observed that $\Pr_{out}\leq 6\times10^{-7}$ is even not achievable if no energy cooperation. Additionally, as can be seen, with the increase of $\Pr_{out,0}$, the above EE gains decrease. This is because less energy from other users is needed  if $\Pr_{out,0}$ increases. Correspondingly, the energy cooperation advantages fade.}

For ``Depleted Energy" scheme adopted  in \cite {MGL}, the harvested energy at every user is used up within every transmission period. The EE gaps between ``No Energy Transferring"  and ``Depleted Energy" schemes show the gains from energy scheduling among different transmission periods. Specifically, around $37\%$, $153\%$, $360\%$ and $484\%$ EE gains can be obtained when $\Pr_{out,0}=10^{-6}$, $10^{-5}$, $5\times10^{-5}$, $10^{-4}$, respectively. It is clear that the gains increase with ${\Pr}_{out,0}$ since in larger ${\Pr}_{out,0}$ case, less energy is needed for data transmission and more energy shall be saved. Depleting energy will definitely lead to a lower EE.

On the other hand, in the ``Uniform Power policy", all sources transmit with the same power obtained by averaging all the harvested energy in $K=10$ transmission periods among the two users. Note for comparison, in the ``Uniform Power policy",  the power at relays refers to the  results obtained with our proposed  algorithm $1$. Moreover, the outage probability threshold cannot be guaranteed. Thus the outage probability requirement is removed. It is shown that our algorithm  outperforms the ``Uniform Power policy" scheme. 
 
To conclude, network coding, energy scheduling among different transmission periods and energy cooperation among different users can provide a notable EE improvement.
 
\section{Conclusions}

We have studied the energy harvesting and wireless energy transferring networks that was coded over finite field. Energy management including determining the optimal power and energy cooperation policies was conducted over  consecutive transmission periods and under the independent but not necessarily identically distributed (i.n.i.d.) Nakagami-$m$ channel environment. The energy efficiency  was maximized under the constraints of the energy causality and outage probability constraints.  With the geometric and fractional programming, the optimization problem was converted into a convex one. The efficient interior-point method was applied to achieve close-to-optimal solutions. The gap between our optimal policy and the decode-and-forward relaying scenario showed the notable energy efficiency gains from the network coding. Additionally, for the network coding scenario, our suboptimal policy outperformed  the   ``No Energy Transferring", ``Depleted Energy"   and  ``Uniform Power" policies. It was shown that the harvested energy was not necessarily depleted and part of the energy was saved for usage in the later transmission periods or  transferred to its cooperative partners. Finally, it was revealed that both the increase of the transmission distance in the first hop and wireless power transmission inefficiency resulted in a degraded energy efficiency.

\section*{Appendix}\label{appendixconvexproof}

We first prove the convexity of the objective function. In every iteration of $q$,
\begin{align}
V'&=M{\alpha_0}{\sum\limits_{k = 1}^K }{{\Pr}_{out,k}}+q{E_{tot}}+q{K }T\sum\limits_{j = 1}^N {{c_j}}\nonumber\hfill\\
&=M{\alpha_0}{ }{\sum\limits_{k = 1}^K }{{\Pr}_{out,k}} +q\cdot\nonumber\hfill\\
&\sum\limits_{k = 1}^K{\left( {\sum\limits_{i = 1}^M {{e^{{{\tilde p}_{i,k}}}T} + (1 - \eta )\sum\limits_{i = 1}^M \sum\limits_{i' = 1}^M {{E_{i \to i',k}}}  +  \sum\limits_{j = 1}^N {{c_j}{e^{{{\tilde p'}_{j,k}}}}T} } } \right)}\nonumber\\
& >0 .\label{Vprime}
\end{align}
The first item in \eqref{Vprime} is the sum of multiple exponential terms multiplied by positive constants and thus convex \cite{24}. Meanwhile, the second item is obvious convex. Then $V'$ is convex. The proof of the convexity property of \eqref{Poutconstr} follows the same approach.

For constraint \eqref{eq63}, we separate it into two inequations, i.e.,

\begin{equation}
\sum\limits_{l = 1}^k {{e^{{{\tilde p}_{i,l}}}T}}  + \sum\limits_{l = 1}^k {\sum\limits_{i' = 1}^M {{E_{i \to i',l}}} }  - \eta \sum\limits_{l = 1}^k {\sum\limits_{i' = 1}^M {{E_{i' \to i,l}}} }  \leqslant \sum\limits_{l = 1}^k {E{u_{i,l}}}, \label{constraintcausality} \end{equation}

\begin{equation}
0< {{e^{{{\tilde p}_{i,k}}}}}\leq p_{max}\label{constraintcausality2}.
\end{equation}

\eqref{constraintcausality} is  convex \emph{w.r.t } to $\mathcal{P}^*$ and ${\mathbf{E}_{I \to I',k}}$  ($\forall k$), due to the fact that the first item in the left side of \eqref{constraintcausality} is convex and the other items are linear. It is obvious that \eqref{eq631a} and \eqref{constraintcausality2} are convex.

The convexity of $\mathbf{P3}$ is proved.

\balance

\end{document}